\documentclass[10pt,twocolumn,twoside]{IEEEtran}
\usepackage[left=0.75in,right=0.75in,top=0.73in,bottom=.76in]{geometry}
\usepackage[utf8]{inputenc}
\usepackage[T1]{fontenc}
\usepackage{url}
\usepackage{ifthen}
\usepackage{multicol}
\usepackage{cite}

\usepackage{graphicx,epstopdf,amssymb,amsthm}
\usepackage{ulem,color}
\usepackage{epstopdf}
\usepackage[utf8]{inputenc}
\usepackage[english]{babel}
\usepackage{caption}
\usepackage{algorithm}
\usepackage{algpseudocode}
\usepackage{amsmath,amsfonts,amssymb,amsthm}
\usepackage{booktabs}
\usepackage{lipsum}

\usepackage[dvipsnames]{xcolor}

\newtheorem{theorem}{Theorem}
\newtheorem{problem}{Question}
\newtheorem{lemma}{Lemma}

\newtheorem{corollary}{Corollary}

\newtheorem{remark}{Remark}

\newcommand{\diag}{\text{diag}}

\newtheorem{definition}{Definition}

\usepackage{cite}
\usepackage[T1]{fontenc}
\usepackage{color}
\usepackage{amsmath}
\usepackage{amsthm}
\usepackage{amssymb}
\usepackage{stmaryrd}
\usepackage{graphicx}
\usepackage{esint}

\makeatletter

\begin{document}


\title{\vspace{0.25in}Distributed Reproduction Numbers of \\ Networked Epidemics}

\author{
Baike She, 
Philip E. Par\'{e}, 
and 
Matthew Hale*
\thanks{*Baike She and Matthew Hale are with the Department of Mechanical and Aerospace Engineering at University of Florida.
Their work was supported by DARPA under award no. HR00112220038; Philip E. Par\'{e} is with the
Elmore Family School of Electrical and Computer Engineering at Purdue University. His work was supported in part 
   by the National Science Foundation, grant NSF-ECCS \#2032258.
E-mails: \{shebaike, matthewhale\}@ufl.edu; philpare@purdue.edu.
}

}

\maketitle


\begin{abstract}
Reproduction numbers are widely used for the estimation and prediction of epidemic spreading processes over networks. 
However, reproduction numbers do not enable estimation and prediction in individual communities within networks,
and they can be difficult to compute due to the aggregation of infection data that is required to do so. 
Therefore, in this work we propose a novel concept of \textit{distributed reproduction numbers} to capture the spreading behaviors of each entity in the network, and we show how to compute
them using certain parameters in networked $SIS$ and $SIR$ epidemic models. 
We use distributed reproduction numbers to derive new conditions under which an outbreak can occur. 
These conditions are then used to derive new conditions for the existence, uniqueness, and stability of equilibrium states. 
Finally, in simulation we use synthetic infection data to illustrate 
how distributed reproduction numbers provide more fine-grained analyses of networked spreading processes than ordinary reproduction numbers.
\end{abstract}

\vspace{-3ex}
\section{Introduction}
\label{intro}
Reproduction numbers are some of the most critical quantities in infectious disease epidemiology \cite{van2017reproduction}.
Reproduction numbers are often among the quantities most urgently estimated for
emerging infectious diseases in outbreak situations, since it is easy for policy-makers to explain the reproduction numbers to the general public,
and their values can also be used to design control
interventions during an established pandemic \cite{soltesz2020effect}. There are two types of most commonly used reproduction numbers. \textit{Basic reproduction numbers} describe the number of secondary infected cases generated by one infected case in the full susceptible population. Meanwhile, \textit{effective reproduction numbers} capture the number of secondary infected cases generated by one infected case in a mixed susceptible and infected population \cite{van2017reproduction}. The critical value of a reproduction number is~$1$, since epidemic spreading processes will exhibit different spreading behaviors if the reproduction number 
is less than or greater than~$1$. 

In control theory,  researchers have used reproduction numbers to model, analyze, and design control strategies for epidemic mitigation problems \cite{she2021network,pascal2022nonsmooth, smith2021convex,casella2020can}.
Recent work has constructed threshold conditions at~$1$ to 
analyze the transient and steady-state behaviors of spreading processes 
based on basic and effective reproduction numbers being less than or greater than~$1$
\cite{mei2017epidemics_review,pare2020modeling_review,zino2021analysis,nowzari2016epidemics}. 
Researchers have also extended the idea of constructing threshold conditions at~$1$ 
from the classic networked $SIS$ models \cite{van2011n} to some recent novel networked spreading models, e.g., networked bi-virus models \cite{bivirus} and coupled networked spreading models \cite{she2021network}. These threshold conditions are typically in terms of the reproduction numbers of the overall network spreading processes. 

However, networked compartmental models often exhibit high heterogeneity in both spreading parameters and local network structures. 
Hence, it can be challenging to use the reproduction numbers of the overall network to characterize the 
spreading behaviors of individual entities within the network. 
For instance, the spreading behavior of COVID-19 in different regions was different across the United States, in terms of infection growth, peak infection date, etc.~\cite{ihme2021modeling}. 

In addition, different regions have different ways of collecting and representing infection data. Along with privacy issues, it is arduous to 
aggregate data from many communities to construct network-level reproduction numbers.
Therefore, while existing network-level reproduction numbers capture the overall epidemic spread in a network, we 
require new reproduction numbers for each 
entity in a network that can be locally computed 
and characterize local spreading. 

To achieve this goal, we first introduce the 
classic networked $SIS$ and $SIR$ models and their reproduction numbers, which we refer to
as ``network-level reproduction numbers.'' 
Then, we propose a group of novel distributed reproduction numbers to
capture the spreading behavior of individual communities. We develop threshold conditions 
as the function of these
distributed reproduction numbers
to study the classic networked $SIS$ and $SIR$ models. In simulation, we illustrate that the distributed reproduction numbers can be estimated locally via synthetic data. 
This simulation also shows that distributed reproduction numbers can be used to infer spreading trends within and between communities, which network-level reproduction numbers cannot do.

To summarize, our contributions are: 
 \begin{itemize}
   \item We introduce a new group of distributed reproduction numbers. Unlike the network-level reproduction number, 
   we illustrate the distributed reproduction numbers can capture spreading behaviors of individual entities within the network (e.g., individual communities);
    \item We show that not only can 
    distributed reproduction numbers capture epidemic spread within individual entities, but 
    can
    also 
    facilitate the
    analysis of the transient and steady-state behaviors of overall networked spreading;
    \item We leverage synthetic infection data to estimate distributed reproduction numbers, and 
    we illustrate that distributed reproduction numbers capture more detailed spreading properties within and between the entities in a network than
    network-level reproduction numbers.
\end{itemize}
The rest of the paper is organized as follows. We introduce the background and problem statements in Section~II. 
In Section~III, we propose and study distributed reproduction numbers. 
In Section~IV, we illustrate the effectiveness of the distributed reproduction numbers using synthetic infection data. 
Then Section~V concludes. 
\vspace{-3ex}
\subsection*{Notation}
Let $\underline{n}$ denote the index set $\{1,2,3,\dots,n\}$. 
For a matrix $\mathcal{A}\in\mathbb{R}^{n\times n}$, we use $[\mathcal{A}]_{ij}$ to denote the $ij^{th}$ entry of $\mathcal{A}$. 
We use $\rho(\mathcal{A})$ to represent the spectral radius of the matrix $\mathcal{A}$. For a vector $x\in \mathbb{R}^n$, we use $\diag(x)\in\mathbb{R}^{n\times n} $ to denote the diagonal matrix with the $i$th diagonal entry being $x_i$ for all $i \in\underline{n}$. For two vectors $x, y\in\mathbb{R}^n$, we use $x>y$ to denote that there exists at least one $i\in\underline{n}$ such that $x_i>y_i$. Denote $\boldsymbol{0}$ and $\boldsymbol{1}$ as the zero vector and one vector with the corresponding dimension given by context. 
Let $[a,b]^n$ denote a closed cube and $(a,b)^n$ denote an open cube,  for any $a,b\in \mathbb{R}$.
\vspace{-2ex}
\section{Problem Formulation}\label{section2}
This section provides background on two existing networked spreading models and the standard
threshold conditions for their network-level reproduction numbers.  
Then we formulate the problems that are the focus of this work.
\vspace{-3ex}
\subsection{Background: $SIR$ and $SIS$ Models}
The networked $SIS$ and $SIR$ models are popular in modeling and analyzing  epidemic spreading processes \cite{mei2017epidemics_review,pare2020modeling_review,zino2021analysis,nowzari2016epidemics}, and we 
study them throughout the paper. 
We consider epidemic processes on strongly connected graphs of~$n$ communities. For all $i\in \underline{n}$, 
let $s_i$, $x_i$, and~$r_i$ represent the susceptible, infected, and recovered proportions of the population of community $i$, respectively. We use $s_i(t)$ and $s_i$ interchangeably.
The classic networked $SIS$ model is 
\begin{subequations}\label{Eq: SIS}
\begin{alignat}{2}
\frac{ds_i}{dt} &= -\sum_{j\in \underline{n}}s_i\beta_{ij}x_j+\gamma_{i}x_i, \label{eq:SIS_S}\\
\frac{dx_i}{dt} &= \sum_{j\in \underline{n}}s_i\beta_{ij}x_j-\gamma_{i}x_i,
\label{eq:SIS_I}
\end{alignat}
\label{eq:SIS}
\end{subequations}
and the classic networked $SIR$ model is 
\begin{subequations}\label{Eq:SIR}
\begin{alignat}{2}
\frac{ds_i}{dt} &= -\sum_{j\in \underline{n}}s_i\beta_{ij}x_j, \label{eq:SIR_S}\\
\frac{dx_i}{dt} &= \sum_{j\in \underline{n}}s_i\beta_{ij}x_j-\gamma_{i}x_i,\label{eq:SIR_I}\\
\frac{dr_i}{dt} &= \gamma_{i}x_i.\label{eq:SIR_R}
\end{alignat}
\label{eq:SIR}
\end{subequations}
In these models, $\beta_{ij}\geq 0$ denotes the transmission rate from community $j$ to community $i$ for all $i,j\in\underline{n}$, and ${\gamma_{i}>0}$ denotes the recovery rate of community $i$ for all $i\in\underline{n}$. Further, we define $\mathcal{B}\in \mathbb{R}^{n\times n}_{\geq0}$ such that $[\mathcal{B}]_{ij}=\beta_{ij}$ for all $i,j \in \underline{n}$ as the \textit{transmission matrix}. We define the diagonal matrix $\mathcal{D}\in \mathbb{R}^{n\times n}_{\geq0}$ such that $[\mathcal{D}]_{ii}=\gamma_i>0$ for all $i \in \underline{n}$ as the \textit{recovery matrix}. 
\begin{definition}
\label{def:Qqui}
At any healthy (disease-free) equilibrium, $x^*=\boldsymbol{0}$. At any endemic equilibrium, $x^*\in (0,1)^n$.
\end{definition}
Note that at an endemic equilibrium one cannot have $x^*_i=0$ or $x_i^*=1$ for any~$i \in \underline{n}$ \cite{mei2017epidemics_review}.
%
Inspired by the use of reproduction numbers to construct threshold conditions for the analysis of non-networked epidemic models, 
researchers have studied the spreading behaviors of the networked $SIS$ and $SIR$ models in \eqref{eq:SIS} and \eqref{eq:SIR} using thresholds conditions \cite{mei2017epidemics_review,pare2020modeling_review,zino2021analysis,nowzari2016epidemics}. 
These threshold conditions are defined in terms of the \textit{reproduction numbers of networks} (namely the network-level reproduction numbers).
\begin{definition}(Reproduction Numbers of Networks)
\label{Def:Repro}
Define the basic reproduction numbers of the networked $SIS$ and $SIR$ models as $R^0= \rho ({\mathcal{D}^{-1}\mathcal{B}})$. Define the effective reproduction numbers of the networked $SIS$ and $SIR$ models  as $R^t = \rho ({\diag(s) \mathcal{D}^{-1}\mathcal{B}})$.
\end{definition}
Further, we can leverage both reproduction numbers of networks in Definition~\ref{Def:Repro}
to characterize the spreading behaviors of the network, illustrated by the following two lemmas.
\begin{lemma}\cite[Thm. 4.2 , 4.3]{mei2017epidemics_review}
\label{lem:SIS}   
The networked $SIS$ model in \eqref{eq:SIS} has a unique equilibrium which is the globally asymptotically stable healthy equilibrium if and only if $R^0\leq 1$. The $SIS$ model has a unique endemic equilibrium that is globally asymptotically stable if and only if $R^0 >1$. 
\end{lemma}
It is less useful to analyze the existence of equilibria of the networked $SIR$ model, since $SIR$ models have an infinite number of 
healthy equilibria but
cannot have any endemic equilibria. Instead, we are more interested in studying the transient behavior of the $SIR$ model. Thus, we use $w$ to denote
the normalized left eigenvector of the matrix $\mathcal{D}-\mathcal{B}$, and the following lemma summarizes the transient behavior of the networked $SIS$  and $SIR$ models.
\begin{lemma}\cite[Thm. 5.2 , 5.4]{mei2017epidemics_review}
\label{lem:SIR}
The weighted average $w^{\top}x$ is decreasing if and only if $R^t< 1$. The weighted average $w^{\top}x$ is increasing if and only if $R^t>1$. The weighted average $w^{\top}x$ remains unchanging if and only if $R^t = 1$.
\end{lemma}
Note that the original theorems \cite[Thm. 4.2, 4.3]{mei2017epidemics_review} and \cite[Thm. 5.2, 5.4]{mei2017epidemics_review} capture the case of homogeneous transmission rates in a network. However, the results 
for heterogeneous transmission networks listed in Lemma~\ref{lem:SIS} and Lemma~\ref{lem:SIR} can be obtained through the similar proofs given by \cite[Thm. 4.2, 4.3]{mei2017epidemics_review} and \cite[Thm. 5.2, 5.4]{mei2017epidemics_review}.
\vspace{-2.5ex}
\subsection{Motivation and Problem Statements}
Lemma~\ref{lem:SIS} and Lemma~\ref{lem:SIR} characterize the spreading behaviors of the networked models in \eqref{eq:SIS} and \eqref{eq:SIR} through the \textit{reproduction numbers of the network}, i.e., $R^0$ and $R^t$.
However, these reproduction numbers of the network may fail to capture the spreading behavior of individual entities within the network. 
To illustrate this point, 
Fig.~\ref{fig_inf} presents the infected proportion of each community in a networked $SIR$ model over ten communities in a strongly connected graph. 
The dashed line indicates $w^{\top}x$. Fig.~\ref{fig_eff_r} shows the corresponding effective reproduction number of the network, i.e., $R^t$. The effective reproduction number $R^t>1$ 
until roughly timestep~$20$. However, the infected proportions of most communities, including $x_1$, $x_2$, and $x_6$, already start decreasing by timestep~$20$. 
\begin{figure}
  \begin{center}
    \includegraphics[ trim = 0cm 0cm 0cm 0cm, clip, width=\columnwidth]{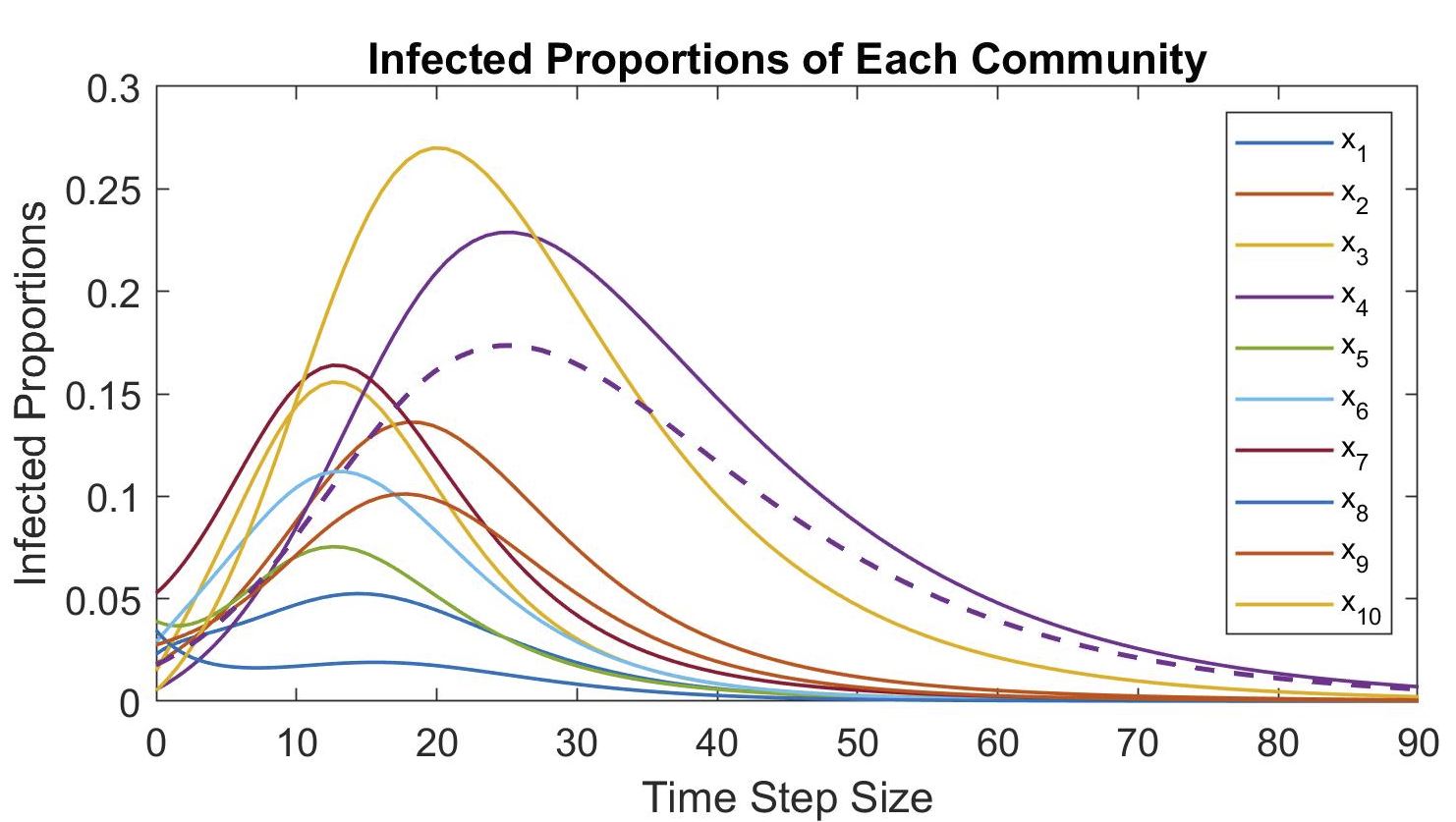}
  \end{center}
    \vspace{-2ex}
  \caption{Infected Proportions of Each Community}
  \label{fig_inf}
  \vspace{-2ex}
\end{figure}
\begin{figure}
  \begin{center}
    \includegraphics[ trim = 2.2cm 1.2cm 2.2cm 0cm, clip, width=\columnwidth]{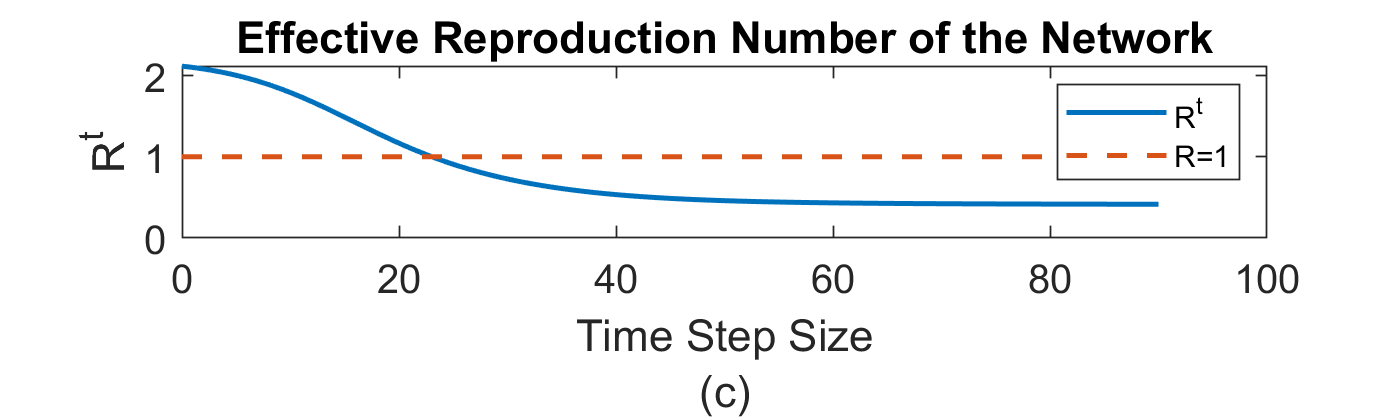}
  \end{center}
    \vspace{-2ex}
  \caption{Effective Reproduction Number of the Networked Epidemic}
    \vspace{-3ex}
  \label{fig_eff_r}
\end{figure}
Therefore,
if we aim to analyze an individual community or a subnetwork of connected communities within the network, the reproduction numbers of the network, $R^0$ and/or $R^t$, may fail to capture the spreading behavior,
since both thresholds in Lemmas~\ref{lem:SIS} and ~\ref{lem:SIR} characterize network-level spreading only.

Motivated by this discussion of the reproduction numbers of networks, we formulate the following problem statements. 
\begin{problem}
\label{prob:1}
How can we define \textit{distributed reproduction numbers} to capture the spreading behaviors within individual communities and between the communities within a network?
\end{problem}
\begin{problem}
\label{prob:2}
Compared to the network-level reproduction numbers of a network,
can we use the distributed reproduction numbers to study the spreading processes at both the individual-level and the network-level?
\end{problem}
\begin{problem}
\label{prob:3}
When capturing epidemic spreading processes in application, what are the advantages of leveraging the distributed reproduction numbers 
developed?
\end{problem}
We will answer Questions~\ref{prob:1}-\ref{prob:3} in the next two sections.
\vspace{-2ex}
\section{Distributed Reproduction Numbers}
\label{section3}
In this section, we  define distributed reproduction numbers for the networked $SIS$ and $SIR$ models, which answers Question~$1$.
We will leverage the distributed reproduction numbers to study the transient and steady-state behaviors of the spreading models.
In order to answer Question~\ref{prob:2}, we will bridge the gap between the distributed reproduction numbers and the network-level reproduction numbers 
by showing that the distributed reproduction numbers can capture the spreading behavior at both the individual- and network-levels.
\vspace{-5ex}
\subsection{Definition of Distributed Reproduction Numbers}
One way to study epidemic spreading processes is to leverage reproduction numbers to indicate the change of the infected population (e.g., increasing, decreasing, unchanging.). As indicated in Lemmas~\ref{lem:SIS} and \ref{lem:SIR}, network-level reproduction numbers can capture the overall spreading behavior within a network. 
However, the spreading behavior of an individual community might not be captured by the reproduction numbers of the network.
Thus, based on the intuition that the infected population of community $i$ will increase if the effective reproduction number of community~$i$ is greater than~$1$, and vice versa, we introduce the following definition of distributed reproduction numbers. 
\begin{definition}[Distributed Reproduction Numbers]
\label{Def:NRN}
For each $i \in \underline{n}$,
define $R^{0}_{ii} = \frac{\beta_{ii}}{\gamma_i}$ as the basic reproduction number within community $i$ itself, and
define $R^{0}_{ij} = \frac{\beta_{ij}}{\gamma_i}$  as the basic reproduction number from community $j$ to community $i$ for each $j \in \underline{n}$. 
For each $i \in \underline{n}$, define $R^{t}_{ii} = \frac{s_i\beta_{ii}}{\gamma_i}$ as the effective reproduction number within community $i$, and define $R^{t}_{ij} = \frac{s_i\beta_{ij}}{\gamma_i}$ for each $j \in \underline{n}$ as the pseudo-effective reproduction number from community $j$ 
to community $i$. 
In addition, we define $I_{ij} = \frac{x_j}{x_i}$ with $x_i, x_j\in(0,1]$ as the infection ratio from the infected proportions of community $j$ to community $i$. Then, we define the effective reproduction number from community $j$ 
to community $i$ as $\bar{R}_{ij}=R^{t}_{ij}I_{ij}$.
\end{definition} 

In order to explain the intuition behind Definition~\ref{Def:NRN}, we consider the group compartmental $SIS$ and $SIR$ models with $\beta$ and $\gamma$ being the transmission and recovery rates, respectively \cite{mei2017epidemics_review}. Note that the two 
models admit $\frac{\beta}{\gamma}$ and $\frac{s(t)\beta}{\gamma}$ as the basic and effective reproduction numbers, respectively.
Based on these terms, we then choose to use
 $R^0_{ij}=\frac{\beta_{ij}}{\gamma_{i}}$ and $R^t_{ij}=\frac{s_i(t)\beta_{ij}}{\gamma_{i}}$ for the basic and pseudo-effective reproduction numbers of the infected proportion $x_{ij}$, where $x_{ij}$ denotes the infected proportion in community $i$ generated by the infected proportion in community $j$ for all $i,j\in\underline{n}$. 
Through the definition, we have $\frac{dx_{ij}}{dt}>0$ if and only if $R^t_{ij}>1$, and vice versa. Further, we define the scaled infected proportion $\bar{x}_{ij}=x_{ij}I_{ij}$ for all $i,j\in\underline{n}$. Note that the 
scaled infected proportion $\bar{x}_{ij}$ can be considered as a normalized infected proportion of $\bar{x}_{ij}$ with respect to community $i$. When computing effective reproductions of community $i$, it is necessary to evaluate the infections from different resources at the same scale.
The following lemma shows the use of $\bar{R}^t_{ij}=R^t_{ij}I_{ij}$.
\begin{lemma} \label{lem:x_ij}
The scaled infected proportion generated by community $j$ in community $i$, denoted by $\bar{x}_{ij}$
for all $i,j \in\underline{n}$, are increasing if and only if $\bar{R}^t_{ij}>1$, and they are decreasing
if and only if~$\bar{R}^t_{ij} < 1$. 
\end{lemma}
\begin{proof}
Based on the fact that $\frac{d\bar{x}_{ij}}{dt}=s_i\beta_{ij}x_j-\gamma_ix_i>0$ if and only if 
$\frac{s_i\beta_{ij}}{\gamma_{i}}I_{ij}=R^t_{ij}I_{ij}>1$, and $\bar{R}^t_{ij}=R^t_{ij}I_{ij}$, it immediately follows that,
for all $i,j \in\underline{n}$,  $x_{ij}$ are increasing if and only if $\bar{R}^t_{ij}>1$. We can use the same technique to show that $\bar{x}_{ij}$ is decreasing if and only if $\bar{R}^t_{ij}<1$. 
\end{proof}
Based on Lemma~\ref{lem:x_ij}, 
the pseudo-effective reproduction numbers together with the infection ratios, i.e., $\bar{R}^t_{ij}=R^t_{ij}I_{ij}$,
are used as the effective reproduction number from community $j$ to community $i$ for all $i,j\in\underline{n}$. 
The basic reproduction number captures the situation where $s_i\approx1$ for all $i\in\underline{n}$, and the effective reproduction number captures cases in which $x_i\in (0,1]$ for all $i\in\underline{n}$. Definition~\ref{Def:NRN} proposes distributed reproduction numbers by separating the infected cases generated in community $i$ in two ways: 
(i) the new cases that are generated through the infected cases within the community itself, and (ii) the new cases that are generated through the infected cases from 
neighboring communities. Hence, we use two types of reproduction numbers, namely the reproduction numbers within a community ($R^0_{ii}$ and $R^t_{ii}$) and the reproduction numbers from one community to another community ($R^0_{ij}$ and $\bar{R}^t_{ij}$), to capture the two types of infection processes. In addition, similar to the reproduction numbers of group compartmental models, we have $R^t_{ii}=s_iR^0_{ii}$ within community $i$ for all $i\in\underline{n}$. 
For the pseudo-effective reproduction numbers from community $j$ to $i$, we have $R^t_{ij}=s_iR^0_{ij}$ for all  $i,j\in\underline{n}$. 

Definition~\ref{Def:NRN} and Lemma~\ref{lem:x_ij} propose a way of explaining spread processes through the distributed reproduction numbers within and between communities. For the purpose of characterizing the spread process of a community, we further define the basic reproduction number and effective reproduction number of a community within the network, through the distributed reproduction numbers in Definition~\ref{Def:NRN}.
\begin{definition}[Reproduction Numbers of Community $i$]
\label{def:R_c}
For all $i \in \underline{n}$, let $R^0_{i}$ denote the basic reproduction number of community $i$, and let $R^t_{i}$ denote the effective reproduction number of community $i$, where
\begin{align} 
   R^0_{i} &= \sum_{j=1}^{n} R^{0}_{ij}, \label{Net_R0} \\
   \bar{R}_i^t &= \sum_{j=1}^{n}\bar{R}_{ij}^t   =\sum_{j=1}^{n} R^t_{ij}I_{ij}. \label{Net_R}
\end{align}
\end{definition}
\begin{remark}
The reproduction numbers defined in \eqref{Net_R0} and \eqref{Net_R} quantify the relationship between the infection within community $i$ and the infections in other communities. Specifically, the basic and effective reproduction numbers of a community are built upon the distributed reproduction numbers from Definition~\ref{Def:NRN}.  Eq.~\eqref{Net_R0} indicates that the basic reproduction number of community $i$ within the network is the sum of the basic reproduction number  within the community $i$ itself and the basic reproduction numbers introduced by its neighbors. Similarly, Eq.~\eqref{Net_R} indicates that 
the effective reproduction number of community $i$ within the network is the sum of the effective reproduction number within 
community $i$ itself and the 
effective reproduction numbers introduced by its neighbors. Further, 
the 
effective reproduction numbers introduced by its neighbors 
are scaled by the infection ratio $I_{ij}$.
For instance, if community $i$ has a lower infected proportion than community $j$ (i.e., $x_i<x_j$), 
then the effective reproduction number from community $j$ to community $i$ will be scaled up by $I_{ij}$. Hence, the effective reproduction number of community $i$ ($\bar{R}^t_i$) can be high, even if the effective reproduction number within community $i$ ($R^t_{ii}$) and the pseudo-effective reproduction numbers from community $j$ to community $i$ ($R^t_{ij}$) are low, since the weights $I_{ij}$ that are also critical can be large. 
\end{remark}
\vspace{-4ex}
\subsection{Properties of Distributed Reproduction Numbers}
Through the distributed reproduction numbers introduced in Definition~\ref{Def:NRN} and Lemma \ref{lem:x_ij}, we can compose the reproduction numbers of an individual community through the sum of the distributed reproduction numbers, as shown in Definition~\ref{def:R_c}. 
Compared to the effective reproduction number of the network ($R^t$), the effective reproduction number of an individual community ($\bar{R}^t_i$) can facilitate the study of the spreading behavior of community $i$ for all $i \in \underline{n}$.
\begin{theorem}
\label{thm:net_r}
When the infected population in community $i$ is nonzero, i.e., $x_i(t)>0$, 
the effective reproduction number $\bar{R}^t_i>1$ if and only if the infected proportion $x_i$ increases; $\bar{R}^t_i<1$ if and only if $x_i$ decreases; $\bar{R}^t_i = 1$ if and only if $x_i$ remains unchanged.
\end{theorem}

\begin{proof}
We show the first statement since the proof of the other statements can follow the same procedure.

$\Leftarrow:$ Recall the definition of the effective reproduction number of community $i$, namely $\bar{R}_i^t = \sum_{j=1}^{n} R^t_{ij}I_{ij}$, for all $i \in \underline{n}$. 
Hence, $\bar{R}_i^t>1$ gives $\bar{R}_i^t = \sum_{j=1}^{n} R^t_{ij}I_{ij}>1$. Then, through Definition~\ref{Def:NRN}, it is true that $\sum_{j=1}^{n}\frac{s_i\beta_{ij}x_j}{x_i\gamma_i}>1$, which leads to $\frac{dx_i}{dt}=\sum_{j=1}^{n} s_i\beta_{ij}x_j-\gamma_ix_i>0$, for all $i\in \underline{n}$. Hence,~$x_i$ is increasing. 

$\Rightarrow:$ If the infected proportion of community $i$ increases,  then $\frac{dx_i}{dt}=\sum_{j=1}^{n} s_i\beta_{ij}x_j-\gamma_ix_i>0$, 
and re-arranging terms immediately gives $\bar{R}_i^t=\sum_{j=1}^{n}\frac{s_i\beta_{ij}x_j}{x_i\gamma_i}>1$, since we have $x_i(t)>0$.
\end{proof}

Theorem~\ref{thm:net_r} demonstrates that the definition of the effective reproduction numbers exhibits thresholding behavior, 
and thus that we can leverage the effective reproduction numbers of communities to capture the spreading behavior within
and between them. 
In addition, recall that the effective reproduction number of the whole network, $R^t$, is monotonically non-increasing, since for all~$i \in \underline{n}$, the value of $s_i$ is monotonically non-increasing \cite{mei2017epidemics_review}. However, 
for all~$i \in \underline{n}$ the value of 
$\bar{R}_i^t$ can be non-monotonic.
\begin{lemma}
\label{lem:non}
For all~$i \in \underline{n}$, the effective reproduction number
$\bar{R}_i^t$ of community~$i$ can be non-monotonic.
\end{lemma}
\begin{proof}
The effective reproduction number $\bar{R}_i^t$ is a weighted sum of $R^t_{ij}$ with the weights $I_{ij}$. Based on Definition~\ref{Def:NRN}, $R^t_{ij}$ is monotonically decreasing, since $s_i(t)$ is monotonically decreasing.
Further, the weights $I_{ij}$ are determined by the ratio between the infected proportions of community $j$ and community $i$ 
for all $i,j\in \underline{n}$. The weights $I_{ij}$ can be non-monotonic, and therefore,
for all~$i \in \underline{n}$, 
the effective reproduction number $\bar{R}_i^t$ 
can be non-monotonic.
\end{proof}

Theorem~\ref{thm:net_r}, Lemma~\ref{lem:x_ij}, and Lemma~\ref{lem:non} demonstrate that we can leverage the distributed reproduction numbers to capture spreading behaviors of individual communities. Hence, we have answered Question~\ref{prob:1} from Section~II. In order to answer Question~\ref{prob:2}, we connect distributed reproduction numbers to the network-level reproduction numbers of networks, 
namely $R^0$ and $R^t$. First we define the distributed reproduction number matrices.
\begin{definition}[Distributed Reproduction Number Matrices]
\label{def:MNR}
The distributed basic and effective reproduction number matrices are
\vspace{-2ex}
\begin{gather}
\mathcal{R}^0 = 
\begin{bmatrix}
         R^0_{11} & R^0_{12} & \cdots & R^0_{1n}\\
         R^0_{21} & R^0_{22} & \cdots & R^0_{2n}\\ 
         \vdots & \vdots & \ddots & \vdots\\ 
         R^0_{n1} & R^0_{n2} & \cdots & R^0_{nn} 
     \end{bmatrix},
     \label{eq:R_0}
\end{gather} 

\vspace{-1ex}

\noindent
and 
\vspace{-1ex}
\begin{gather}
\mathcal{R}^t = 
\begin{bmatrix}
         R^t_{11} & R^t_{12} & \cdots & R^t_{1n}\\
         R^t_{21} & R^t_{22} & \cdots & R^t_{2n}\\ 
         \vdots & \vdots & \ddots & \vdots\\ 
         R^t_{n1} & R^t_{n2} & \cdots & R^t_{nn} 
     \end{bmatrix},\label{eq:R_t}
\end{gather} 

\vspace{-1ex}

\noindent
respectively.
\end{definition}
\begin{remark}
The distributed basic reproduction number matrix
$\mathcal{R}^0=\mathcal{D}^{-1}\mathcal{B}$ is the \textit{next generation matrix \cite{diekmann2010construction}} of the networked $SIS$/$SIR$ models. Thus, the distributed effective reproduction number matrix $\mathcal{R}^t$ is equal to~$\diag(s(t))\mathcal{R}^0$. 
However, the advantage of viewing $\mathcal{R}^0$ and $\mathcal{R}^t$ as the composition of distributed reproduction numbers in \eqref{eq:R_0} and \eqref{eq:R_t} is that we can construct these matrices through the distributed reproduction numbers directly from data.
For instance, in real-world epidemic spreading processes, when we need the network-level effective reproduction number $\rho(\mathcal{R}^t)$, 
instead of estimating the model parameters $\beta_{ij}$, $\gamma_{i}$, and $s_i(t)$ to obtain $\mathcal{R}^0$ and $\mathcal{R}^t$, we only need the estimated distributed reproduction numbers $R^t_{ij}$ to compose $\mathcal{R}^t$. 
We will further illustrate this idea in Section~\ref{sec:simulation}. 
\end{remark}
Based on Definition~\ref{Def:Repro}, it can be observed that $\mathcal{R}^0=\mathcal{D}^{-1}\mathcal{B}$ and $\mathcal{R}^t=\diag(s)\mathcal{D}^{-1}\mathcal{B}$. Hence, the spectral radius of the distributed basic reproduction number matrix,
denoted $\rho(\mathcal{R}^0)$, is the basic reproduction number of the network, i.e., $\rho(\mathcal{R}^0)=R^0$. Meanwhile, the spectral radius of the distributed effective reproduction number matrix, denoted $\rho(\mathcal{R}^t)$, is the effective reproduction number of the network, i.e., $\rho(\mathcal{R}^t)=R^t$. Further, for all~$i \in \underline{n}$,
the $i^{th}$ row sum of the distributed basic reproduction number matrix is the reproduction number of community $i$, 
i.e., we have $\sum_{j=1}^n [\mathcal{R}^0]_{ij}=R^0_i$. 
Note that the  $i^{th}$ row sum of the distributed effective reproduction number matrix, $\mathcal{R}^t$, is not equal to $R_i^t$, since the weights $I_{ij}$ are
not included in~$\mathcal{R}^t$. 

Through studying the spreading behavior of the network, we connect the effective reproduction number of the network to the effective reproduction numbers of the communities.
\begin{theorem}
\label{thm:connection}
When the epidemic states are not at a healthy equilibrium, the following statements hold:
\begin{itemize}
    \item $\bar{R}^t_i=1$ for all $i\in \underline{n}$ only if $\rho(\mathcal{R}^t)=1$;
    \item $\bar{R}^t_i<1$ for all $i\in \underline{n}$ only if $\rho(\mathcal{R}^t)<1$;
    \item $\bar{R}^t_i>1$ for all $i\in \underline{n}$ only if $\rho(\mathcal{R}^t)>1$.
\end{itemize}
\end{theorem}
\begin{proof}
We start by proving the first statement. If $\bar{R}_i^t=1$ and $x_i>0$ for all $i\in \underline{n}$, 
then we have that the matrix $\diag(x(t))^{-1}\mathcal{R}^t\diag(x(t))$ is a row stochastic matrix;
this can be seen by noting that the $i^{th}$ row sum is $\bar{R}_i^t = \sum_{j=1}^{n} R^t_{ij}I_{ij} = 1$. 
Hence, based on the fact that the spectral radius of a row stochastic matrix is~$1$, we have
\begin{equation*}
    \rho\Big(\diag(x(t))^{-1}\mathcal{R}^t\diag(x(t))\Big)=1.
\end{equation*}
Further, it is true that $\rho(\diag(x(t))^{-1}\mathcal{R}^t\diag(x(t)))=\rho(\mathcal{R}^t)=1$,
since similarity transformations preserve eigenvalues, i.e.,
$\diag(x(t))^{-1}\mathcal{R}^t\diag(x(t))$ and $\mathcal{R}^t$ have the same spectrum.
Therefore, if $\bar{R}_i^t=1$ for all $i\in \underline{n}$, 
then we must have $\rho(\mathcal{R}^t) = 1$.

Next we show the second statement. If $\bar{R}^t_i  = \sum_{j=1}^{n} R^t_{ij}I_{ij} <1$ for all $i\in \underline{n}$, then,
using the fact that~$\bar{R}_i^t$ is equal
to the~$i^{th}$ row sum of~$\diag(x(t))^{-1}\mathcal{R}^t\diag(x(t))$, we see that 
we must have $[(\diag(x(t))^{-1}\mathcal{R}^t\diag(x(t))]_{ij}\in[0,1)$ for all $i,j\in \underline{n}$. 
Now suppose for the sake of contradiction 
that $\rho(\diag(x(t))^{-1}\mathcal{R}^t\diag(x(t)))=\rho(\mathcal{R}^t)\geq1$. 
Then, by increasing some non-zero entries of the matrix  $\diag(x(t))^{-1}\mathcal{R}^t\diag(x(t))$ through changing $\mathcal{R}^t$, we can construct a new matrix $\diag(x(t))^{-1}\mathcal{\tilde{R}}^t\diag(x(t))$ such that $\diag(x(t))^{-1}\mathcal{\tilde{R}}^t\diag(x(t))$ is a stochastic matrix,
i.e., its row sums equal~$1$. 
Thus, $\rho (\diag(x(t))^{-1}\mathcal{\tilde{R}}^t\diag(x(t)))=1$. 

Note that the transmission matrix $\mathcal{B}$ is an irreducible matrix since we assume that the graphs that capture the transmission networks are strongly connected. Further, the model parameters $\beta_{ij}$ and $\gamma_{i}$ for all $i,j\in\underline{n}$
along with the infected state $x_i(t)$ for all $i\in\underline{n}$ are positive. Thus, the matrices $\diag(x(t))^{-1}\mathcal{R}^t\diag(x(t))$ and $\diag(x(t))^{-1}\mathcal{\tilde{R}}^t\diag(x(t))$ are nonnegative and irreducible. 

Based on \cite[ Thm. 2.7 and Lemma 2.4]{varga2009matrix_book},
the spectral radius of a non-negative irreducible matrix will increase when any entry of the matrix increases, which gives 
\begin{align*} 
  \rho (\diag(x(t))^{-1}\mathcal{R}^t\diag(x(t)))&<\rho (\diag(x(t))^{-1}\mathcal{\tilde{R}}^t\diag(x(t)))\\
   &=1.
\end{align*}
This result 
contradicts the aforementioned hypothesis that $\rho(\diag(x(t))^{-1}\mathcal{R}^t\diag(x(t)))=\rho(\mathcal{R}^t)\geq1$. Therefore, we must have $\rho(\diag(x(t))^{-1}\mathcal{R}^t\diag(x(t)))=\rho(\mathcal{R}^t)<1$.

We can use the same techniques to show the third statement. Hence, we complete the proof.
\end{proof}
\begin{remark}
Theorem~\ref{thm:connection} bridges the gap between $R^t$ and $\bar{R}^t_i$. Especially, for the case where the overall information of the network is unknown, we can leverage the distributed effective reproduction numbers of each community to indicate the effective reproduction number of the whole network, and further to determine the overall spreading behavior.
\end{remark}
Using Theorem~\ref{thm:connection} and Lemma~\ref{lem:SIS}, we can characterize the spreading behaviors of the $SIS$ and $SIR$
models through~$\bar{R}^t_i$. 

\begin{corollary}
\label{cor:SIS_SIR}
For the networked $SIR$ model,
the healthy equilibria are locally stable if $\bar{R}^t_i<1$ for all $i \in \underline{n}$. For the networked $SIS$ model, if $\bar{R}^t_i>1$
for all $i \in \underline{n}$, then there must exist a unique endemic equilibrium, which is stable.
\end{corollary}
\begin{proof}
If $\bar{R}^t_i<1$ for all $i \in \underline{n}$, then
based on Theorem~\ref{thm:connection} we have $\rho(\mathcal{R}^t)<1$. Further, for 
the networked $SIR$ model, each diagonal entry of $\diag(s(t))$ is monotonically decreasing,
which indicates that $\rho(\mathcal{R}^t)$ is monotonically decreasing (based on \cite[ Thm. 2.7, and Lemma 2.4]{varga2009matrix_book}). 
Based on Lemma~\ref{lem:SIR}, $\rho(\mathcal{R}^t)<1$ will cause the infected population in each community to converge to the healthy 
equilibria, and the healthy equilibria are stable. 
Thus, we have established local stability. 

For the networked $SIS$ model, if $\bar{R}^t_i>1$ for all 
$i\in \underline{n}$, based on Theorem~\ref{thm:connection}, we have $\rho(\mathcal{R}^t)>1$. Further, based on \cite[ Thm. 2.7, and Lemma 2.4]{varga2009matrix_book} and $s(t)<\textbf{1}$, we have $\rho(\mathcal{R}^0)>\rho(\mathcal{R}^t)>1$. Hence, based on Lemma~\ref{lem:SIS}, the networked $SIS$ model will have a unique endemic equilibrium, which is stable, under the condition that $\bar{R}^t_i>1$ for all $i\in \underline{n}$. Thus, we have completed the proof.
\end{proof}
\begin{remark}
Corollary~\ref{cor:SIS_SIR} provides a new way to analyze the spreading behavior of the classic 
$SIS$ and $SIR$ models, i.e., through the distributed reproduction numbers. For networked $SIR$ models, if the reproduction number of every community in the network is less than~$1$, then we can ensure the epidemic is fading away. For networked $SIS$ models, if the effective reproduction number of each community in the network is greater than~$1$, then there must be an endemic in the future.
\end{remark}

Lemmas~\ref{lem:SIS} and~\ref{lem:SIR} indicate that if $R^0<1$, then the weighted sum of the infected states will converge to zero. However, 
unlike the fact that
$R^t\leq R^0$ for network-level reproduction numbers, 
if the basic reproduction number of community $i$ is less than~$1$ for all $i\in \underline{n}$, then the effective reproduction number of 
community $i$ can still be greater than~$1$. Hence, we have the following corollary.

\begin{corollary}
\label{cor：outbreak}
There can be an outbreak within community $i$ even under the condition that the basic reproduction number of the community is smaller than~$1$, 
i.e., $R^0_i<1$ for all $i \in \underline{n}$.
\end{corollary}
\begin{proof}
Using $R^0_i=\sum_{j=1}^{n} R^0_{ij}<1$ for all $i,j\in \underline{n}$, we cannot guarantee that 
$\bar{R}^t_i=\sum_{j=1}^{n} s_iR^0_{ij}I_{ij}<1$ for all $i,j\in \underline{n}$.
Hence, if $R^0_i<1$ but $\bar{R}^t_i>1$,  based on Theorem~\ref{thm:net_r}, the infected proportion $x_i$ will increase, which will cause an outbreak. Hence, we complete the proof.
\end{proof}

After demonstrating that we can use the distributed reproduction numbers to analyze spreading behaviors, we showed that the  distributed 
reproduction numbers are closely related to the basic and effective reproduction numbers of the network in Theorem~\ref{thm:connection}, Corollary~\ref{cor:SIS_SIR}, and Corollary~\ref{cor：outbreak}. Hence we have answered Question~\ref{prob:2} that not only can we use the distributed reproduction numbers to analyze spreading behaviors of individual entities within the network, but also we can use the distributed reproduction numbers to study the overall spreading behavior of the network as a whole.

\vspace{-2ex}
\section{Applications} \label{sec:simulation}
In this section, we use two examples to illustrate the importance of leveraging distributed reproduction numbers to study 
epidemic spread across entire networks and within each entity in a network. 
In the first example, we show the advantage of leveraging $\bar{R}^t_i$ for all $i \in\underline{n}$ instead of $R^t$ in analyzing networked spreading processes. 
In the second example, we illustrate the potential of leveraging distributed reproduction numbers in data-driven applications.
Together, these examples answer Question~3 from Section~II. 
\vspace{-4ex}
\subsection{Using Distributed Effective Reproduction Numbers}
Consider an epidemic spreading over ten strongly connected communities, as shown in Fig.~\ref{fig_ten}. Suppose the epidemic spreads based on the classic networked $SIR$ models in \eqref{eq:SIR}. 
We capture the spreading behavior in the plots in Fig.~\ref{fig:simulation}. 
Fig.~\ref{fig:simulation}~(Top) shows that the effective reproduction number of the whole network ($R^t$) is always less than~$1$. 
Thus, if a community uses 
this $R^t$ for policy-making and forecasting, then the community might believe that the 
infected proportion of the population will decrease from time step zero. 

However, there 
are still outbreaks over several communities, e.g., communities 3 and 5 in 
Fig.~\ref{fig:simulation}~(Bottom), where the infected proportions in fact increase 
for several timesteps at the 
beginning. These outbreaks can be explained through the distributed reproduction numbers of the 
communities, which are plotted in Fig.~\ref{fig:simulation}~(Middle). Through 
analyzing Fig.~\ref{fig:simulation}~(Middle), $\bar{R}^t_3$ and $\bar{R}^t_5$ are 
greater than~$1$ at the beginning, which indicate outbreaks within them
and thus that both communities should take actions against 
potential outbreaks. This simple case demonstrates that it is more informative to leverage 
distributed reproduction numbers when designing mitigation policies for individual community. %
\vspace{-3ex}
\begin{figure}[h!]
  \begin{center}
    \includegraphics[ trim = 8cm 7.2cm 8cm 6cm, clip, width=0.9\columnwidth]{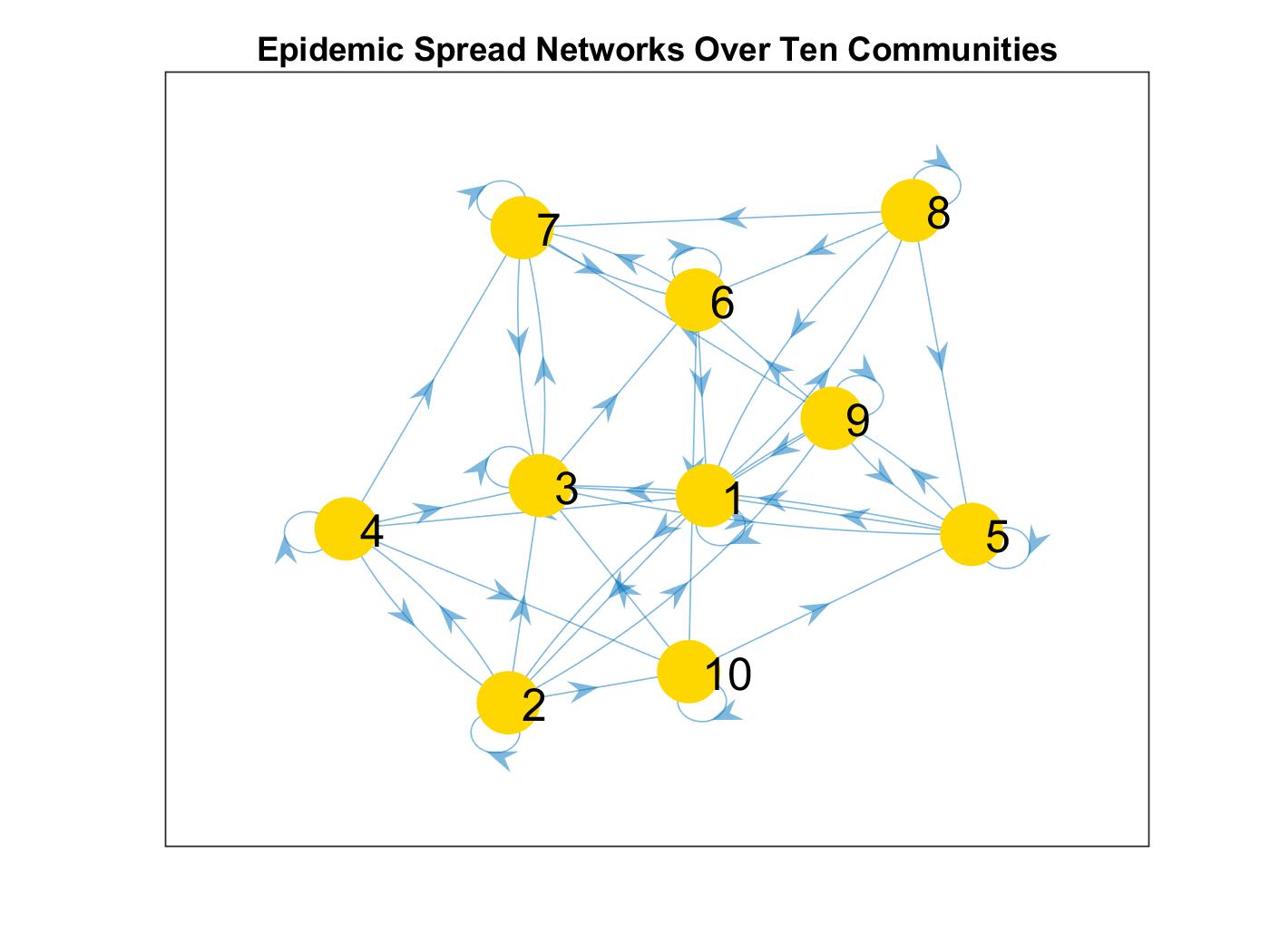}
  \end{center}
    \vspace{-2ex}
  \caption{We consider an $SIR$ model of an epidemic spreading over this strongly connected network of~$10$ communities.}
  \label{fig_ten}
    \vspace{-3ex}
\end{figure}
\begin{figure}
  \begin{center}
    \includegraphics[ trim = 0.3cm 0cm 0cm 0.1cm, clip, width=0.9\columnwidth]{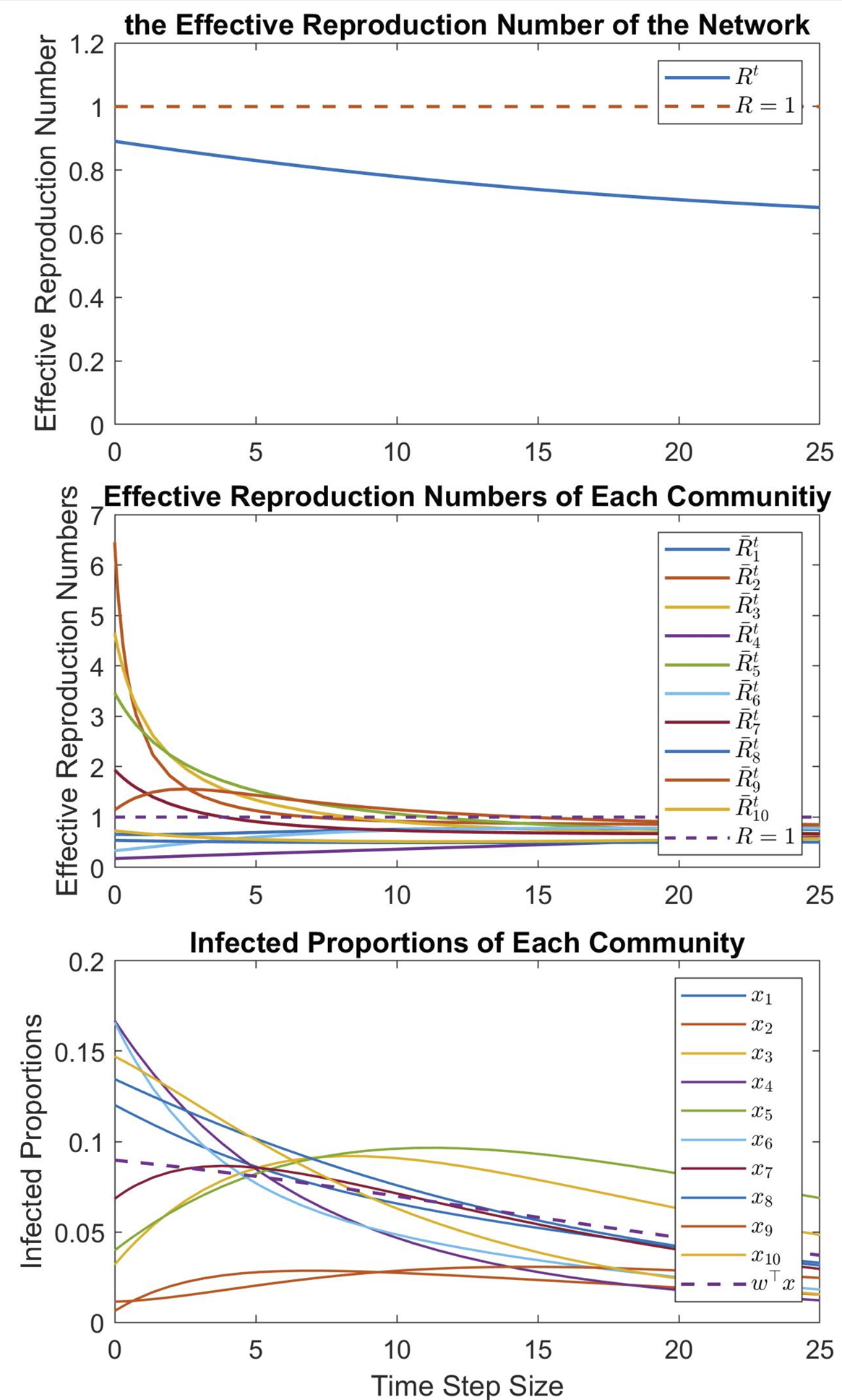}
  \end{center}
\vspace{-2ex}
  \caption{
(Top) The effective reproduction number of the whole spread network. (Middle) The effective reproduction number of each community in the network, (Bottom) The infected proportion of each community. Note that the dashed line is the weighted sum of the infected proportions, and the trend of weighted sum is captured by $R^t$ in the top figure.} 
  \label{fig:simulation}
    \vspace{-3ex}
\end{figure}
\vspace{-2ex}
\subsection{Data-Driven Distributed Reproduction Numbers}
In the second example, we show how we can calculate distributed reproduction numbers from synthetic data to analyze spreading 
over a network, which completes our answer to Question~\ref{prob:3} from Section~II. We construct a piece-wise time-varying stochastic $SIR$ spreading process over the strongly connected 3-node network shown in Fig.~$\ref{fig_3_node}$. 

Suppose each community has a population of~$20,000$. 
We set the initial infected population of communities 1, 2, and 3 as
12, 3, and 23, respectively.
We assume the spread occurs from $2020-03-03$ to $2020-06-02$. For the daily infected cases, we assume that we know the total infected cases and the sources which cause these infections (e.g., new infections caused within a community or caused by neighbors). We use the method proposed by \cite{cori2013new} to estimate the distributed effective reproduction numbers defined in this work, where a statistical inference approach is used and the estimation only relies on daily infected cases.
\vspace{-2ex}
\begin{figure}[h]
  \begin{center}
    \includegraphics[ trim = 7cm 7cm 7cm 4cm, clip, width=0.4\columnwidth]{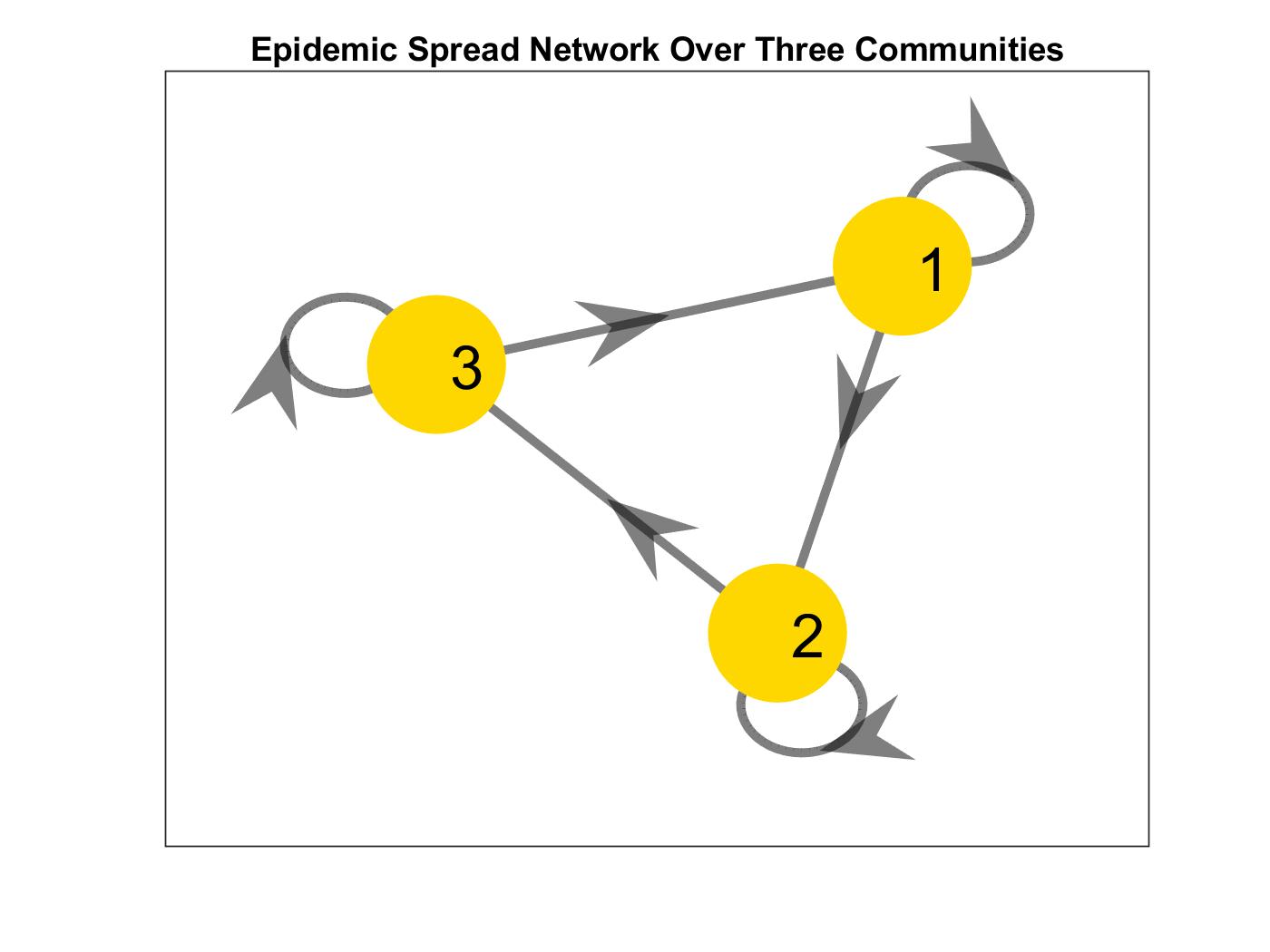}
  \end{center}
    \vspace{-2ex}
  \caption{We consider epidemic spread over the three strongly connected communities shown here.}
  \label{fig_3_node}
\vspace{-1ex}
\end{figure}

Fig.~\ref{fig_R_total}~(Top) shows the total daily infected cases of the three connected communities. We estimated the effective reproduction number of the network $R^t$, 
given in Fig.~\ref{fig_R_total}~(Bottom). The estimated $R^t$ can illustrate that the total infected cases start to decrease around $2020-04-10$, since $R^t<1$ after that date. However, the estimated $R^t$ in Fig.~\ref{fig_R_total} cannot provide any further information about the spreading behaviors of the individual communities. 
\vspace{-2ex}
\begin{figure}[h]
  \begin{center}
    \includegraphics[ trim = 0.0cm 0.cm 0.12cm 0cm, clip, width=\columnwidth]{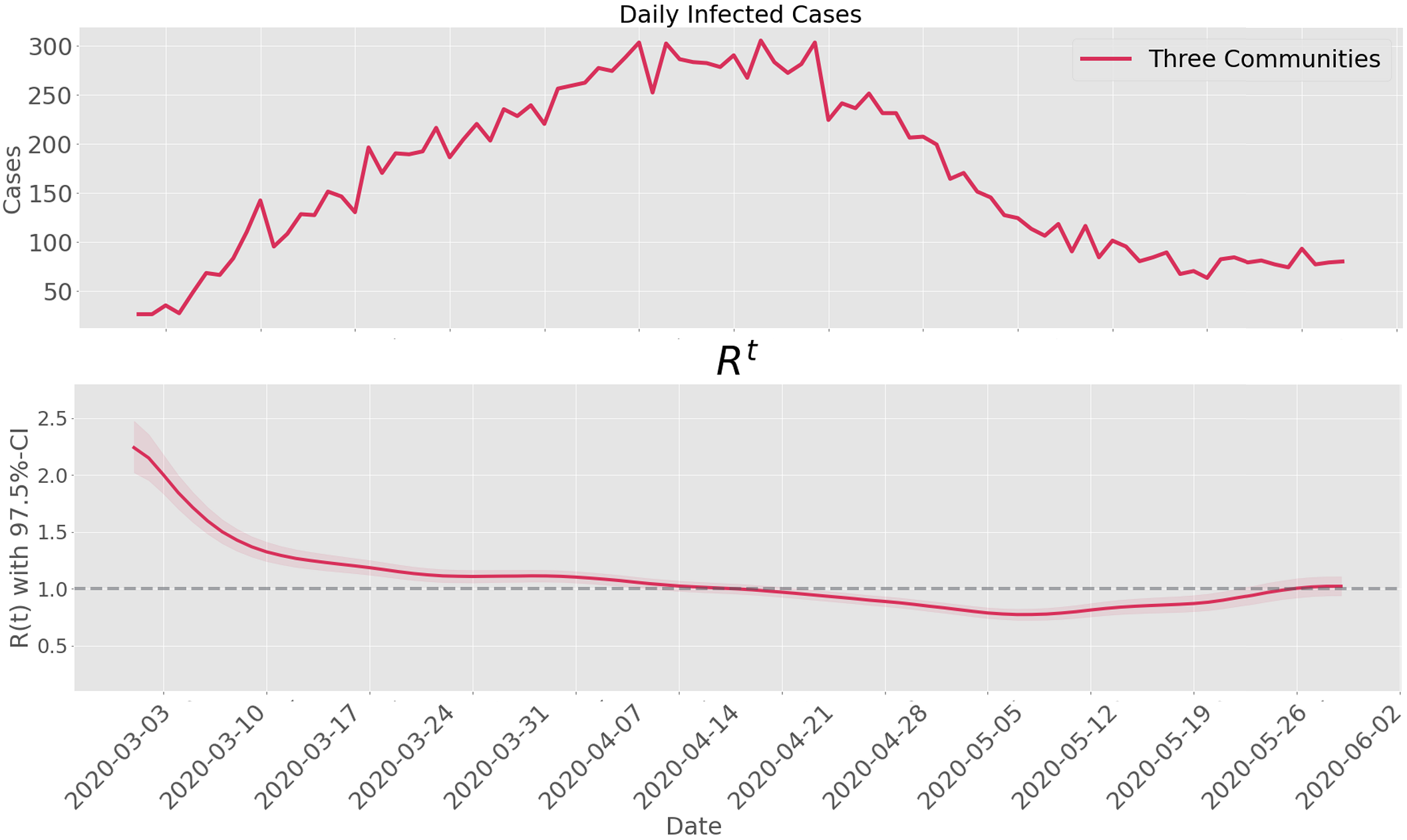}
  \end{center}
    \vspace{-2ex}
  \caption{(Top) Total infected cases over the spreading network, (Bottom) Estimated effective reproduction number~$R^t$.}
  \label{fig_R_total}
   \vspace{-1ex}
\end{figure}

In order to study the spreading behavior of each community, we leverage the daily infected cases of each community to estimate the 
effective
reproduction number of each community, denoted $\bar{R}^t_i$ for $i\in \underline{3}$. We plot the  infected cases of each community, shown in Fig.~\ref{fig_3} (Top). Further, we plot the estimated effective reproduction numbers of communities 1, 2, and 3 ($\bar{R}^t_1$, $\bar{R}^t_2$, and $\bar{R}^t_3$) in Fig.~\ref{fig_3} (Bottom three). As shown by $\bar{R}^t_1\approx1$ between $2020-03-24$ to $2020-04-21$, 
the daily infected cases of community~$1$ remain unchanged. Meanwhile, according to $\bar{R}^t_3$ from $2020-03-17$ to $2020-03-24$, it is possible for the outbreak in community 3 to become worse, since 
$\bar{R}^t_3>1$ and $\bar{R}^t_3$ is increasing. Hence, compared to $R^t$ in Fig.~\ref{fig_R_total}, $\bar{R}^t_i$ for $i\in \underline{3}$ can give more specific insights about individual communities.
\begin{figure}[h]
  \begin{center}
    \includegraphics[trim = 0cm 0cm 0.08cm 0cm, clip, width=\columnwidth]{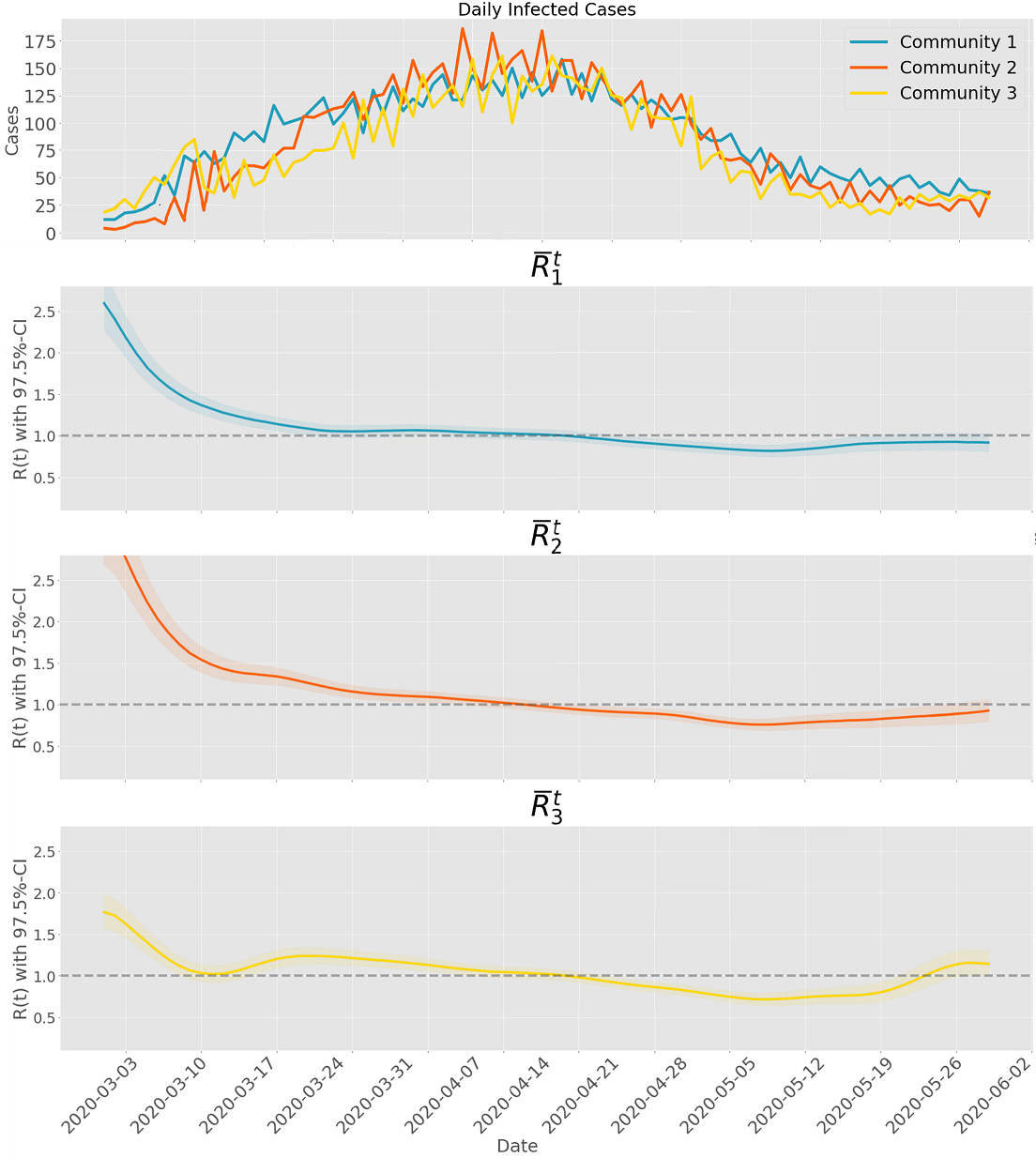}
  \end{center}
    \vspace{-2ex}
  \caption{ (Top) Infected cases of each community within the spreading network, (Top Middle) Estimated effective reproduction number of Community 1, (Bottom Middle) Estimated effective reproduction number of Community 2, (Bottom) Estimated effective reproduction number of Community 3.}
  \label{fig_3}
    \vspace{-3ex}
\end{figure}

Last, given the assumption that we know the cause of the infections of the communities (either generated by the infected 
cases within the community or the infected cases from its neighbors), we can leverage the infected cases generated by
the neighbors (Fig.~\ref{fig_R_ij}~(Top)) to estimate the distributed pseudo-effective reproduction numbers $R^t_{ij}$. 
As shown in Fig.~\ref{fig_R_ij}~(Bottom Middle), through observing the distributed reproduction number $R^t_{21}$,
we find that $R^t_{21}\approx1$ from $2020-03-31$ to $2020-05-05$. This observation indicates that the daily infections
in community $2$ that are introduced from the infected population in community $1$ remain unchanged. 
Thus, if we want to further decrease the infections in community $2$ generated 
by community $1$, we 
can implement 
a mitigation strategy that restricts interactions
between these two communities and monitor $R^t_{21}$ as an indicator of the 
effectiveness of the
strategy. 
For further indications, according to $R^t_{13}$ in Fig.~\ref{fig_R_ij}~(Top Middle), 
$R^t_{13}>1$ from $2020-03-10$ to $2020-03-17$, 
then $R^t_{13}<1$ from $2020-03-17$ to $2020-03-24$. One can infer 
that the interaction from communities $3$ to $1$ may vary drastically from $2020-03-10$ to $2020-03-24$, due to the implementation of mitigation policies between these communities, e.g. traffic restrictions from $2020-03-17$ to $2020-03-24$.
\begin{figure}[h]
  \begin{center}
    \includegraphics[trim = 0.0cm 0.0cm 0cm 0cm, clip, width=\columnwidth]{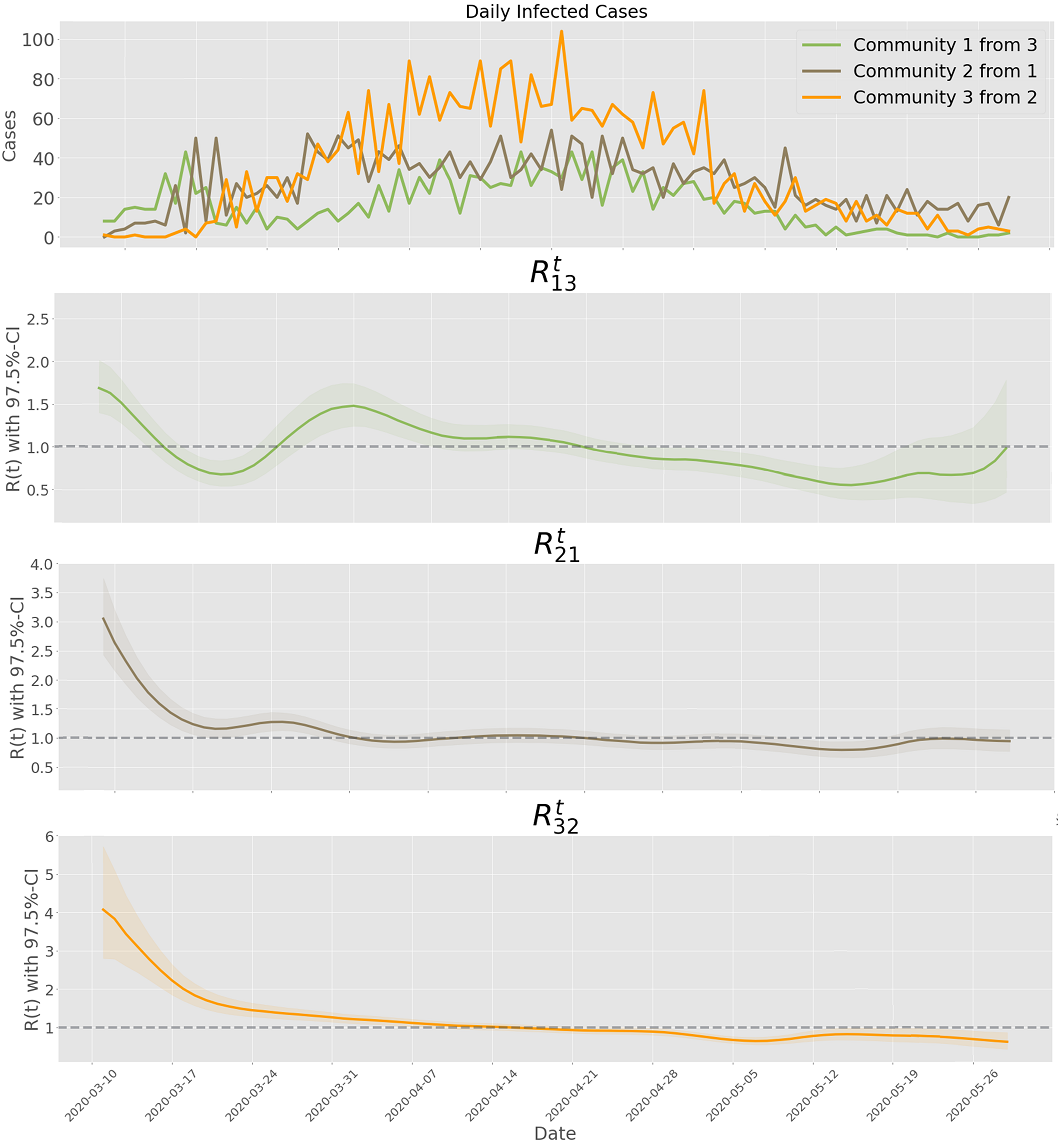}
  \end{center}
    \vspace{-2ex}
  \caption{(Top) Infected cases of each community generated by its neighbors, (Top Middle) Estimated distributed pseudo-effective reproduction number $R^t_{13}$, (Bottom Middle) Estimated distributed pseudo-effective reproduction number $R^t_{21}$, (Bottom) Estimated distributed pseudo-effective reproduction number $R^t_{32}$.}
  \label{fig_R_ij}
    \vspace{-3ex}
\end{figure}
\vspace{-2ex}
\section{Conclusion and Future Work}
In this work, we defined distributed reproduction numbers to study spreading behaviors of networked epidemic models. We bridged the gap between 
these distributed reproduction numbers and the conventional reproduction numbers of networks. In addition, we
demonstrated that distributed reproduction numbers can capture both spreading behaviors of individual communities
and networks as a whole.
Compared to network-level reproduction numbers, we show that distributed reproduction numbers estimated from synthetic data can infer 
much more information about the spread of an epidemic.
Future work will consider how we can leverage real-world networked testing and contact tracing data to estimate distributed reproduction numbers and how we can design distributed interventions by using these estimated distributed reproduction numbers.
\vspace{-2ex}

\normalem
\bibliographystyle{IEEEtran}
\bibliography{IEEEabrv,main}


\end{document}